\documentclass{article}
\usepackage{amsmath,amssymb,amsfonts}

\def\op#1{\mathop{{\it\fam0} #1}\limits}

\newcommand{\beq}{\begin{equation}}
\newcommand{\eeq}{\end{equation}}
\newcommand{\ben}{\begin{eqnarray}}
\newcommand{\een}{\end{eqnarray}}
\newcommand{\be}{\begin{eqnarray*}}
\newcommand{\ee}{\end{eqnarray*}}
\newcommand{\bea}{\begin{eqalph}}
\newcommand{\eea}{\end{eqalph}}

\newcommand{\bL}{{\mathbf L}}
\newcommand{\cD}{{\mathcal D}}

\newcommand{\la}{\lambda}

\newcommand{\m}{\mu}

\newcommand{\vt}{\vartheta}
\newcommand{\cG}{{\mathfrak g}}

\newcommand{\up}{\upsilon}

\newcommand{\si}{\sigma}
\newcommand{\Si}{\Sigma}

\newcommand{\w}{\wedge}

\newcommand{\wt}{\widetilde}

\newcommand{\ol}{\overline}

\newcommand{\dr}{\partial}

\newcommand{\ar}{\op\longrightarrow}

\newcommand{\ot}{\otimes}

\newenvironment{eqalph}{\stepcounter{equation}
\setcounter{equationa}{\value{equation}} \setcounter{equation}{0}

\begin{eqnarray}}{\end{eqnarray}\setcounter{equation}{\value{equationa}}}

\newcounter{equationa}
\newcounter{remark}
\newcounter{example}
\newcounter{theorem}
\newcounter{proposition}
\newcounter{lemma}
\newcounter{corollary}
\newcounter{definition}
\def\theremark{\arabic{remark}}

\def\thedefinition{\arabic{theorem}}

\newenvironment{proof}{{\it Proof.}}{\hfill $\Box$
\medskip }

\newenvironment{theorem}{\refstepcounter{theorem} \medskip{\bf
Theorem \thedefinition.}\it}{\medskip }

\newcommand{\mar}[1]{}

\begin{document}

\hbox{}

\begin{center}

{\Large\bf Theory of Classical Higgs Fields. II. Lagrangians}

\bigskip

G. SARDANASHVILY, A. KUROV

\medskip

Department of Theoretical Physics, Moscow State University, Russia

\bigskip

\end{center}

\begin{abstract}
We consider classical gauge theory with spontaneous symmetry
breaking on a principal bundle $P\to X$ whose structure group $G$
is reducible to a closed subgroup $H$, and sections of the
quotient bundle $P/H\to X$ are treated as classical Higgs fields.
In this theory, matter fields with an exact symmetry group $H$ are
described by sections of a composite bundle $Y\to P/H\to X$. We
show that their gauge $G$-invariant Lagrangian necessarily
factorizes through a vertical covariant differential on $Y$
defined by a principal connection on an $H$-principal bundle $P\to
P/H$ (Theorems 5 and 6).
\end{abstract}

\bigskip

Following our previous work \cite{higgs13}, we consider classical
gauge theory on a principal bundle $P\to X$ with a structure Lie
group $G$ which is reducible to its closed subgroup $H$, i.e., $P$
admits reduced principal subbundles possessing a structure group
$H$.

Given a closed (and, consequently, Lie) subgroup $H\subset G$, we
have a composite bundle
\mar{b3223a}\beq
P\to P/H\to X, \label{b3223a}
\eeq
where
\mar{b3194}\beq
 P_\Si=P\ar P/H \label{b3194}
\eeq
is a principal bundle with a structure group $H$ and
\mar{b3193}\beq
\Si=P/H\ar X \label{b3193}
\eeq
is a $P$-associated bundle with a typical fibre $G/H$ which a
structure group $G$ acts on by left multiplications. In accordance
with the well-known theorem \cite{book09,ste}, there is one-to-one
correspondence between the global sections $h$ of the quotient
bundle (\ref{b3193}) and the reduced $H$-principal subbundles
$P^h$ of $P$ which are the restriction
\mar{001}\beq
P^h=h^*P_\Si \label{001}
\eeq
of the $H$-principal bundle $P_\Si$ (\ref{b3194}) to $h(X)\subset
\Si$. In classical gauge theory, global sections of the quotient
bundle  (\ref{b3193}) are treated as classical Higgs fields
\cite{book09,sard06a,sard14}.

A question is how to describe matter fields in gauge theory with a
structure group $G$ if they admit only an exact symmetry subgroup
$H$. In particular, this is the case of spinor fields in
gravitation theory \cite{iva,sard11}.

We have shown that such matter fields are represented by sections
of a composite bundle
\mar{b3225}\beq
\pi_{YX}: Y\ar \Si\ar X \label{b3225}
\eeq
where $Y\to \Si$ is a $P_\Si$-associated bundle
\mar{bnn}\beq
Y=(P\times V)/H \label{bnn}
\eeq
with a structure group $H$ acting on its typical fibre $V$ on the
left \cite{book09,sard06a,sard14}. Given a global section $h$ of
the fibre bundle $\Si\to X$ (\ref{b3193}), the restriction
\mar{b3226}\beq
Y^h=h^*Y=(h^*P\times V)/H =(P^h\times V)/H \label{b3226}
\eeq
of a fibre bundle $Y\to\Si$ to $h(X)\subset \Si$ is a fibre bundle
associated with the reduced $H$-principal subbundle $P^h$
(\ref{001}) of a $G$-principal bundle $P$. Sections of the fibre
bundle $Y^h\to X$ (\ref{b3226}) describe matter fields in the
presence of a Higgs field $h$. A key point is that the composite
bundle $\pi_{YX}:Y\to X$ (\ref{b3225}) is proved to be a
$P$-associated bundle
\be
Y=(P\times (G\times V)/H)/G
\ee
with a structure group $G$ \cite{book09,higgs13,sard14}. Its
typical fibre is a fibre bundle
\mar{wes}\beq
\pi_{WH}:W=(G\times V)/H \to G/H \label{wes}
\eeq
associated with an $H$-principal bundle $G\to G/H$. A structure
group $G$ acts on the $W$ (\ref{wes}) by the induced
representation \cite{mack}:
\mar{iik}\beq
g: (G\times V)/H \to (gG\times V)/H. \label{iik}
\eeq

This fact enables one to describe matter fields with an exact
symmetry group $H\subset G$ in the framework of gauge theory on a
$G$-principal bundle $P\to X$ if its structure group $G$ is
reducible to $H$. Here, we aim to show that their gauge
$G$-invariant Lagrangian necessarily factorizes through a vertical
covariant differential on $Y$ defined by an $H$-principal
connection on $P\to P/H$ (Theorems \ref{010} and \ref{tmp50}).

A problem is that, though the $P$-associated composite bundle
$Y\to X$ (\ref{b3225} can be endowed with a principal connection
on a $G$-principal bundle $P\to X$, such a connection need not be
reducible to principal connections on reduced $H$-principal
subbundles $P^h$, unless the following condition holds
\cite{book09,kob}.

\begin{theorem} \label{redt} \mar{redt}
Let a Lie algebra $\cG$ of $G$ be a direct sum
\mar{g13}\beq
\cG = {\mathfrak h} \oplus {\mathfrak f} \label{g13}
\eeq
of a Lie algebra ${\mathfrak h}$ of $H$ and its supplement
$\mathfrak f$ obeying the commutation relations
\be
[{\mathfrak f},{\mathfrak f}]\subset {\mathfrak h}, \qquad
[{\mathfrak f},{\mathfrak h}]\subset \mathfrak f.
\ee
(e.g., $H$ is a Cartan subgroup of $G$). Let $A$ be a principal
connection on $P$. The ${\mathfrak h}$-valued component $\ol A^h$
of its pull-back onto a reduced $H$-principal subbundle $P^h$ is a
principal connection on $P^h$.
\end{theorem}

At the same time, connections on reduced $H$-principal subbundles
$P^h$ can be generated in a different way. Let $\Pi\to Z$ be a
principal bundle with a structure group $K$. Given a manifold map
$\phi:Z'\to Z$, a pull-back bundle $\phi^*\Pi\to Z'$ also is a
principal bundle with a structure group $K$. Let $A$ be a
principal connection on a principal bundle $\Pi\to Z$. Then the
pull-back connection $\phi^*A$ is a principal connection on
$\phi^*\Pi\to Z'$ \cite{kob}. The following is a corollary of this
fact.

\begin{theorem} \label{mos178}
Given the composite bundle (\ref{b3223a}), let $A_\Si$ be a
principal connection on the $H$-principal bundle $P\to\Si$
(\ref{b3194}). Then, for any reduced $H$ principal subbundle $P^h$
(\ref{001}) the pull-back  connection $h^*A_\Si$ is a principal
connection on $P^h$.
\end{theorem}

Turn now to the composite bundle $Y$ (\ref{b3225}). Given an atlas
$\Psi_P$ of $P$, the associated quotient bundle $\Si\to X$
(\ref{b3193}) is provided with bundle coordinates $(x^\la,\si^m)$.
With this atlas and an atlas $\Psi_{Y\Si}$ of $Y\to \Si$, the
composite bundle $Y\to X$ (\ref{b3225}) is endowed with adapted
bundle coordinates $(x^\la,\si^m,y^i)$ where $(y^i)$ are fibre
coordinates on $Y\to\Si$. Then the following holds
\cite{book09,sard06a,sard14}.

\begin{theorem} \mar{LL4} \label{LL4}
Let
\mar{510f04}\beq
A_\Si=dx^\la\ot(\dr_\la + A^a_\la e_a) + d\si^m\ot(\dr_m + A^a_m
e_a) \label{510f04}
\eeq
be a principal connection on an $H$-principal bundle $P\to \Si$
where $\{e_a\}$ is a basis for a Lie algebra ${\mathfrak h}$ of
$H$. Let
\mar{510f05}\beq
A_{Y\Si}=dx^\la\ot(\dr_\la + A^a_\la(x^\m,\si^k) I_a^i\dr_i) +
d\si^m\ot(\dr_m + A^a_m (x^\m,\si^k)I^i_a\dr_i) \label{510f05}
\eeq
be an associated principal connection on $Y\to\Si$ where $\{I_a\}$
is a representation of a Lie algebra ${\mathfrak h}$ in $V$. Then,
for any subbundle $Y^h\to X$ (\ref{b3226}) of a composite bundle
$Y\to X$, the pull-back connection
\mar{510f06}\beq
A_h=h^*A_{Y\Si}= dx^\la\ot[\dr_\la + (A^a_m (x^\m,h^k)\dr_\la h^m
+A^a_\la(x^\m,h^k)) I_a^i\dr_i], \label{510f06}
\eeq
is a connection on $Y^h$ associated with the pull-back principal
connection $h^*A_\Si$ on the reduced $H$-principal subbundle $P^h$
in Theorem \ref{mos178}.
\end{theorem}

Any connection $A_\Si$ (\ref{510f04}) on a fibre bundle $Y\to\Si$
yields a first order differential operator, called the vertical
covariant differential,
\beq
\wt D: J^1Y\to T^*X\op\otimes_Y V_\Si Y, \qquad \wt D=
dx^\la\otimes(y^i_\la- A^i_\la -A^i_m\si^m_\la)\dr_i, \label{7.10}
\eeq
on a composite bundle $Y\to X$ where $V_\Si Y$ is the vertical
tangent bundle of $Y\to\Si$. It possesses the following important
property \cite{book09,sard06a,sard14}.

\begin{theorem} \label{tmp40}
For any section $h$ of a fibre bundle $\Si\to X$, the restriction
of the vertical differential $\wt D$ (\ref{7.10}) onto the fibre
bundle $Y^h$ (\ref{b3226}) coincides with a covariant differential
$D^{A_h}$  defined by the connection $A_h$ (\ref{510f06}) on
$Y^h$.
\end{theorem}

In view of Theorems \ref{LL4} and \ref{tmp40}, one can assume that
a Lagrangian of matter fields represented by sections of the
composite bundle (\ref{b3225}) factorizes through the vertical
covariant differential $\wt D$ (\ref{7.10}) of some connection
$A_{Y\Si}$ on a fibre bundle $Y\to \Si$. Forthcoming Theorem
\ref{010} shows that this factorization is necessary in order that
a matter field Lagrangian to be gauge invariant.

\begin{theorem} \label{010} \mar{010}
In gauge theory on a principal bundle $P$ whose structure group
$G$ is reducible to a closed subgroup $H$, a matter field
Lagrangian is gauge invariant only if it factorizes through a
vertical covariant differential of some principal connection on
the  $H$-principal bundle $P\to P/H$ (\ref{b3194}).
\end{theorem}

\begin{proof}
Let $P\to X$ be a principal bundle whose structure group $G$ is
reducible to a closed subgroup $H$. Let $Y$ be the
$P_\Si$-associated bundle (\ref{bnn}). A total configuration space
of gauge theory of principal connections on $P$ in the presence of
matter and Higgs fields is
\mar{510f00}\beq
J^1C\op\times_X J^1Y \label{510f00}
\eeq
where $C=J^1P/G$ is the bundle of principal connections on $P$ and
$J^1Y$ is the first order jet manifold of $Y\to X$. A total
Lagrangian on the configuration space (\ref{510f00}) is a sum
\mar{510f01}\beq
L_{\rm tot}=L_{\rm A} +L_{\rm m} + L_\si \label{510f01}
\eeq
of a gauge field Lagrangian $L_{\rm A}$, a matter field Lagrangian
$L_{\rm m}$ and a Higgs field Lagrangian $L_\si$. The total
Lagrangian $L_{\rm tot}$ (\ref{510f01}) is required to be
invariant with respect to vertical principal automorphisms of a
$G$-principal bundle $P\to X$.  Any vertical principal
automorphism of a $G$-principal bundle $P\to X$, being
$G$-equivariant, also is $H$-equivariant and, thus, it is a
principal automorphism of an $H$-principal bundle $P\to\Si$.
Consequently, it yields an automorphism of the $P_\Si$-associated
bundle $Y$ (\ref{b3225}). Accordingly, every $G$-principal vector
field $\xi$ on $P\to X$ (an infinitesimal generator of a local
one-parameter group of vertical principal automorphisms of $P$)
also is an $H$-principal vector field on $P\to \Si$. It yields an
infinitesimal gauge transformation $\up_\xi$ of a composite bundle
$Y$ seen as a $P$- and $P_\Si$-associated bundle. This reads
\mar{rty}\beq
\up_\xi= \xi^p(x^\m)J_p^m\dr_m + \vt_\xi^a(x^\m,\si^k)I_a^A\dr_A,
\label{rty}
\eeq
where $\{J_p\}$ is a representation of a Lie algebra $\cG$ of $G$
in $G/H$ and $\{I_a\}$ is a representation of a Lie algebra
$\mathfrak h$ of $H$ in $V$. Since gauge and Higgs field
Lagrangians in the absence of matter fields are assumed to be
gauge invariant, a matter field Lagrangian $L_{\rm m}$ also is
separately gauge invariant.  This means that its Lie derivative
along the jet prolongation $J^1\up_\xi$ of the vector field
$\up_\xi$ (\ref{rty}) vanishes, that is,
\mar{510f03}\beq
\bL_{J^1\up_\xi}L_{\rm m}=0. \label{510f03}
\eeq
In order to satisfy the conditions (\ref{510f03}), let us consider
some principal connection $A_\Si$ (\ref{510f04}) on an
$H$-principal bundle $P\to\Si$ and the associated connection
$A_{Y\Si}$ (\ref{510f05}) on $Y\to\Si$. Let a matter field
Lagrangian $L_{\rm m}$ factorize as
\be
L_{\rm m}:J^1Y\op\to^{\wt D}T^*X\op\ot_Y V_\Si Y\to\op\w^nT^*X
\ee
through the vertical covariant differential $\wt D$ (\ref{7.10}).
In this case, $L_{\rm m}$ can be regarded as a function $L_{\rm
m}(y^i,k^i_\la)$ of formal variables $y^i$ and $k^i_\la=\wt
D^i_\la$. The corresponding infinitesimal gauge transformation of
variables $(y^i,k^i_\la)$ reads
\be
\up=\up^i\dr_i +\dr_j\up^i k^j_\la\frac{\dr}{\dr k^i_\la}.
\ee
It is independent of derivatives  of gauge parameters $\xi$.
Therefore, the gauge invariance condition is trivially satisfied.
\end{proof}

However, a problem is that the principal connection $A_\Si$
(\ref{510f04}) on an $H$-principal bundle $P\to P/H$ fails to be a
dynamic variable in gauge theory. Therefore, let us assume that a
Lie algebra of a structure group $G$ satisfies the decomposition
(\ref{g13}). In this case, any $G$-principal connection $A$ on a
principal bundle $P\to X$ yields $H$-principal connections on
reduced $H$-principal subbundles $P^h$ in accordance with Theorem
\ref{redt}. Then one can state the following \cite{sard14}.

\begin{theorem} \label{tmp50} \mar{tmp50}
There exists a connection $A_{Y\Si}$ (\ref{510f05}) on a fibre
bundle $Y\to P/H$ whose restriction $A_h=h^*A_\Si$ onto a
$P^h$-associated bundle $Y^h$ coincides with a principal
connection $\ol A_h$ on $P^h$ generated by a principal connection
$A$ on a principal bundle $P\to X$.
\end{theorem}

\begin{proof} Let $P^h\subset P$ be a reduced principal subbundle and $\ol A_h$ an $H$-principal
connection on $P^h$ in Theorem  \ref{redt} which is generated by a
$G$-principal connection $A$ on a principal bundle $P\to X$. It is
extended to a $G$-principal connection on $P$ so that $h$ is an
integral section of the associated connection
\be
\ol A_h=dx^\la\ot(\dr_\la +  A_\la^p J^m_p\dr_m)
\ee
on a $P$-associated bundle $\Si\to X$. Let $\Psi_{Y\Si}$ be an
atlas of a $P_\Si$-associated bundle $Y\to\Si$ which is defined by
a family $\{z_\iota\}$ of local sections of $P\to\Si$. Given a
section $h$ of $\Si\to X$, we have a family of sections
$\{z_\iota\circ h\}$ which yields an atlas $\Psi^h$ of a principal
bundle $P\to X$ with $H$-valued transition functions. With respect
to this atlas, a section $h$ takes its values in the center of a
quotient space $G/H$ and the connection $\ol A_h$ reads
\beq
\ol A_h=dx^\la\ot(\dr_\la + A_\la^a e_a). \label{gyu}
\eeq
We have
\mar{00a}\beq
A=\ol A_h +\Theta= dx^\la\ot(\dr_\la + A_\la^a e_a) +\Theta_\la^b
dx^\la\ot e_b, \label{00a}
\eeq
where $\{e_a\}$ is a basis for the Lie algebra ${\mathfrak h}$ and
$\{e_b\}$ is that for ${\mathfrak m}$. Written with respect to an
arbitrary atlas of $P$, the decomposition (\ref{00a}) reads
\be
A=\ol A_h +\Theta, \qquad \Theta=\Theta_\la^p dx^\la\ot e_p,
\ee
and obeys the relation
\be
\Theta_\la^p J_p^m=\nabla^A_\la h^m,
\ee
where $D_\la$ are covariant derivatives relative to the associated
principal connection $A$ on $\Si\to X$. Based on this fact, let
consider the covariant differential
\be
D=D^m_\la dx^\la\ot\dr_m=(\si_\la^m- A^p_\la J_p^m)dx^\la\ot\dr_m
\ee
relative to the associated principal connection $A$ on $\Si\to X$.
It can be regarded as a $V\Si$-valued one-form on the jet manifold
$J^1\Si$ of $\Si\to X$. Since the decomposition (\ref{00a}) holds
for any section $h$ of $\Si\to X$, there exists a $(VP/P)$-valued
(where $VP$ is the vertical tangent bundle of $P\to X$) one-form
\be
\Theta=\Theta_\la^p dx^\la\ot e_p
\ee
on $J^1\Si$ which obeys the equation
\mar{nmb}\beq
\Theta_\la^p J_p^m=D_\la^m. \label{nmb}
\eeq
Then we obtain the $(VP/G)$-valued one-form
\mar{plk}\be
A_H=dx^\la\ot(\dr_\la + (A_\la^p -\Theta_\la^p)e_p) \label{plk}
\ee
on $J^1\Si$ whose pull-back onto each $J^1h(X)\subset J^1\Si$ is
the connection $\ol A_h$ (\ref{gyu}) written with respect to the
atlas $\Psi^h$. The decomposition (\ref{00a}) holds and,
consequently, the equation (\ref{nmb}) possesses a solution for
each principal connection $A$. Therefore, there exists a
$(VP)/G$-valued one-form
\mar{ljl}\beq
A_H=dx^\la\ot(\dr_\la + (a^p_\la- \Theta_\la^p)e_p) \label{ljl}
\eeq
on the product $J^1\Si\op\times_X J^1C$ such that, for any
principal connection $A$ and any Higgs field $h$, the restriction
of $A_H$ (\ref{ljl}) to
\be
J^1h(X)\times A(X)\subset J^1\Si\op\times_X J^1C
\ee
is the connection $\ol A_h$ (\ref{gyu}) written with respect to
the atlas $\Psi^h$. Let us now assume that, whenever $A$ is a
principal connection on a $G$-principal bundle $P\to X$, there
exists a principal connection $A_\Si$ (\ref{510f04}) on a
principal $H$-bundle $P\to \Si$ such that the pull-back connection
$A_h=h^*A_{Y\Si}$ (\ref{510f06}) on $Y^h$ coincides with $\ol A_h$
(\ref{gyu}) for any $h\in\Si(X)$. In this case, there exists
$V_\Si Y$-valued one-form
\mar{510fh}\beq
\wt \cD=dx^\la\ot(y^i_\la - (A^a_m\si_\la^m +A^a_\la)I_a^i)\dr_i
\label{510fh}
\eeq
on the configuration space (\ref{510f00}) whose components are
defined as follows. Given a point
\mar{lvv}\beq
(x^\la, a^r_\m,a^r_{\la\m}, \si^m,\si^m_\la,y^i,y^i_\la) \in
J^1C\op\times_X J^1Y, \label{lvv}
\eeq
let $h$ be a section of $\Si\to X$ whose first jet $j^1_xh$ at
$x\in X$ is $(\si^m,\si^m_\la)$, i.e.,
\be
h^m(x)=\si^m, \qquad \dr_\la h^m(x)=\si^m_\la.
\ee
Let the bundle of principal connections $C$ and the Lie algebra
bundle $VP/G$ be provided with the atlases associated with the
above mentioned atlas $\Psi^h$. Then we write
\mar{lvv1}\beq
A_h=\ol A_h, \qquad A^a_m\si_\la^m +A^a_\la = a^a_\la-
\Theta_\la^a. \label{lvv1}
\eeq
These equations for functions $A^a_m$ and $A^a_\la$ at the point
(\ref{lvv}) have a solution because $\Theta_\la^a$ are affine
functions in the jet coordinates $\si_\la^m$.
\end{proof}

Given solutions of the equations (\ref{lvv1}) at all points of the
configuration space (\ref{510f00}), we require that a matter field
Lagrangian factorizes as
\mar{lvv2}\beq
L_{\rm m}:J^1C\op\times_XJ^1Y\op\to^{\wt D}T^*X\op\ot_Y V_\Si
Y\to\op\w^nT^*X \label{lvv2}
\eeq
through the form $\wt D$ (\ref{510fh}). As a result, we obtain a
gauge theory of gauge potentials of a group $G$, matter fields
with an exact symmetry subgroup $H\subset G$ and classical Higgs
fields on the configuration space (\ref{lvv}).

As was mentioned above, an example of classical Higgs fields is a
metric gravitational field in gauge gravitation theory on natural
bundles with spontaneous symmetry breaking caused by the existence
of Dirac spinor fields with the exact Lorentz symmetry group
\cite{iva,sard11}. Describing spinor fields in terms of the
composite bundle (\ref{b3225}), we get their Lagrangian
(\ref{lvv2}) in the presence of a general linear connection which
is invariant under general covariant transformations
\cite{book09,sard11}. Classical gauge fields also are considered
in gauge theory on gauge-natural bundles \cite{pal} and in Stelle
-- West gravitation theory \cite{lecl}.

Let us note however that the symmetry breaking mechanism of
Standard Model differs from that we consider here. Matter fields
in Standard Model admit a total group of symmetries which are
broken because of the existence of a background Higgs vacuum field
\cite{nov}.

\end{document}